\documentclass[sigconf,noacm]{acmart}
\usepackage{bbding}

\usepackage[T1]{fontenc} 
\usepackage[utf8]{inputenc} 
\usepackage[english]{babel}
\usepackage{stfloats} 
\usepackage{booktabs,subcaption,amsfonts,dcolumn}
\usepackage{soul}
\usepackage{paralist}
\usepackage{enumitem}
\usepackage{textcase} 
\usepackage{filemod} 
\usepackage{etoolbox} 
\usepackage{mathtools}
\usepackage{flafter} 
\makeatletter
\def\@copyrightspace{\relax}
\makeatother
\setcopyright{none}
\renewcommand\footnotetextcopyrightpermission[1]{} 
\makeatletter
\renewcommand\@formatdoi[1]{\ignorespaces}
\makeatother

\usepackage{amsthm}
\theoremstyle{plain}

\newtheorem*{theorem*}{Theorem}
\usepackage{float}
\usepackage{multirow}
\makeatletter
\usepackage{tabularx}

\numberwithin{equation}{section}
\allowdisplaybreaks

\def\@noindentfalse{\global\let\if@noindent\iffalse}
\def\@noindenttrue {\global\let\if@noindent\iftrue}
\def\@aftertheorem{%
  \@noindenttrue
  \everypar{%
    \if@noindent%
      \@noindentfalse\clubpenalty\@M\setbox\z@\lastbox%
    \else%
      \clubpenalty \@clubpenalty\everypar{}%
    \fi}}

\theoremstyle{plain}
\newtheorem{theorem}{Theorem}[section]
\AfterEndEnvironment{theorem}{\@aftertheorem}
\newtheorem{definition}[theorem]{Definition}
\AfterEndEnvironment{definition}{\@aftertheorem}

\AfterEndEnvironment{lemma}{\@aftertheorem}

\AfterEndEnvironment{corollary}{\@aftertheorem}

\theoremstyle{definition}

\AfterEndEnvironment{remark}{\@aftertheorem}

\AfterEndEnvironment{example}{\@aftertheorem}
    
\AfterEndEnvironment{proof}{\@aftertheorem}

\setcopyright{none}
\renewcommand\footnotetextcopyrightpermission[1]{}
\usepackage[linewidth=1pt]{mdframed}
\let\@@todo\todo
\def\todo#1{\@@todo[color=red,backgroundcolor=red!10,size=\tiny]{#1}}

\usepackage{delimset}
\def\given{\mskip 0.5mu plus 0.25mu\vert\mskip 0.5mu plus 0.15mu}
\newcounter{bracketlevel}%
\def\@bracketfactory#1#2#3#4#5#6{%
\expandafter\def\csname#1\endcsname##1{%
\global\advance\c@bracketlevel 1\relax%
\global\expandafter\let\csname @middummy\alph{bracketlevel}\endcsname\given%
\global\def\given{\mskip#5\csname#4\endcsname\vert\mskip#6}\csname#4l\endcsname#2##1\csname#4r\endcsname#3%
\global\expandafter\let\expandafter\given\csname @middummy\alph{bracketlevel}\endcsname%
\global\advance\c@bracketlevel -1\relax%
}%
}
\def\bracketfactory#1#2#3{%
\@bracketfactory{#1}{#2}{#3}{relax}{0.5mu plus 0.25mu}{0.5mu plus 0.15mu}
\@bracketfactory{b#1}{#2}{#3}{big}{1mu plus 0.25mu minus 0.25mu}{0.6mu plus 0.15mu minus 0.15mu}
\@bracketfactory{bb#1}{#2}{#3}{Big}{2.4mu plus 0.8mu minus 0.8mu}{1.8mu plus 0.6mu minus 0.6mu}
\@bracketfactory{bbb#1}{#2}{#3}{bigg}{3.2mu plus 1mu minus 1mu}{2.4mu plus 0.75mu minus 0.75mu}
\@bracketfactory{bbbb#1}{#2}{#3}{Bigg}{4mu plus 1mu minus 1mu}{3mu plus 0.75mu minus 0.75mu}
}
\bracketfactory{clc}{\lbrace}{\rbrace}
\bracketfactory{clr}{(}{)}
\bracketfactory{cls}{[}{]}
\bracketfactory{abs}{\lvert}{\rvert}
\bracketfactory{norm}{\Vert}{\Vert}
\bracketfactory{floor}{\lfloor}{\rfloor}
\bracketfactory{ceil}{\lceil}{\rceil}
\bracketfactory{angle}{\langle}{\rangle}

\let\original@left\left
\let\original@right\right
\renewcommand{\left}{\mathopen{}\mathclose\bgroup\original@left}
\renewcommand{\right}{\aftergroup\egroup\original@right}

\newcounter{ctr}\loop\stepcounter{ctr}\edef\X{\@Alph\c@ctr}%
	\expandafter\edef\csname s\X\endcsname{\noexpand\mathscr{\X}}
	\expandafter\edef\csname c\X\endcsname{\noexpand\mathcal{\X}}
	\expandafter\edef\csname b\X\endcsname{\noexpand\boldsymbol{\X}}
	\expandafter\edef\csname I\X\endcsname{\noexpand\mathbb{\X}}
\ifnum\thectr<26\repeat

\let\@IE\IE\let\IE\undefined
\newcommand{\IE}{\mathop{{}\@IE}\mathopen{}}
\let\@IP\IP\let\IP\undefined
\newcommand{\IP}{\mathop{{}\@IP}}

\def\^#1{\relax\ifmmode {\mathaccent"705E #1} \else {\accent94 #1}\fi}
\def\~#1{\relax\ifmmode {\mathaccent"707E #1} \else {\accent"7E #1}\fi}
\def\*#1{\relax#1^\ast}
\edef\-#1{\relax\noexpand\ifmmode {\noexpand\bar{#1}} \noexpand\else \-#1\noexpand\fi}
\def\>#1{\vec{#1}}
\def\.#1{\dot{#1}}

\def\atop{\@@atop}

\renewcommand{\leq}{\leqslant}
\renewcommand{\geq}{\geqslant}
\renewcommand{\phi}{\varphi}

\newcommand\indep{\protect\mathpalette{\protect\@indep}{\perp}}
\def\@indep#1#2{\mathrel{\rlap{$#1#2$}\mkern2mu{#1#2}}}

\def\parsetime#1#2#3#4#5#6{#1#2:#3#4}
\def\parsedate#1:20#2#3#4#5#6#7#8+#9\empty{20#2#3-#4#5-#6#7 \parsetime #8}
\def\moddate{\expandafter\parsedate\pdffilemoddate{\jobname.tex}\empty}

\makeatother

\usepackage{algorithm}
\usepackage{algpseudocode}
\usepackage{float}
\usepackage{url}
 \usepackage{hyperref} 

\usepackage{amsmath}

\author{De Zhang Lee}
\affiliation{%
  \institution{National University of Singapore}
  \country{Singapore}
}

\author{Han Fang}
\affiliation{%
  \institution{National University of Singapore}
  \country{Singapore}
}

\author{Ee-Chien Chang}
\affiliation{%
  \institution{National University of Singapore}
  \country{Singapore}
}

\title{Enhancing Stateful Detection of Adversarial Attacks with Soft-labels' Temporality and Robust Similarity Approximations}
\begin{abstract}

Stateful Detection (SD) mitigates adversarial attacks by determining whether a sequence of  queries  contains  queries from a black-box adversary. Recent works, such as Blacklight and  PIHA utilize query similarity to detect such queries. In this paper,  we observe that temporal information, in particular, the temporal correlation of the classification soft labels,  is a prominent characteristic of adversarial attacks and can be leveraged to reduce false positive rates.
Moreover, we  point out a potential vulnerability in SD implementation.  Many SD systems identify similar queries according  to some implicit, computationally expensive  metric.  To improve efficiency,  these systems often adopt an approximate similarity function as substitute. This discrepancy could be exploited  by  crafting queries that appear dissimilar under the approximation but are close in the intended metric, thereby evading detection. We refer to this as an ``adversarial attack'' on the approximation function, and 
demonstrate  it   through a lightweight attack on Blacklight's similarity function.

Based on the above observations, we propose  a two-phase approach. The first phase  identifies  subsequences of queries with high similarity,  incorporating   randomness  to prevent the aforementioned ``adversarial attacks''. The second phase  analyzes temporal correlation of the soft-labels to further validate the presence of the adversary's queries. 
Experimental results show that the framework detects adversarial queries generated by Boundary Attack, HSJA, SimBA, Square Attack with true positive rate (TPR) reaching 1.00, while maintaining a false positive rate (FPR) of at most 0.06. Additionally, the method is  robust against OARS which is an adaptive attack. 
\end{abstract}

\begin{document}

\maketitle


\section{Introduction}


\begin{figure*} 
    \centering
    \includegraphics[width=\linewidth]{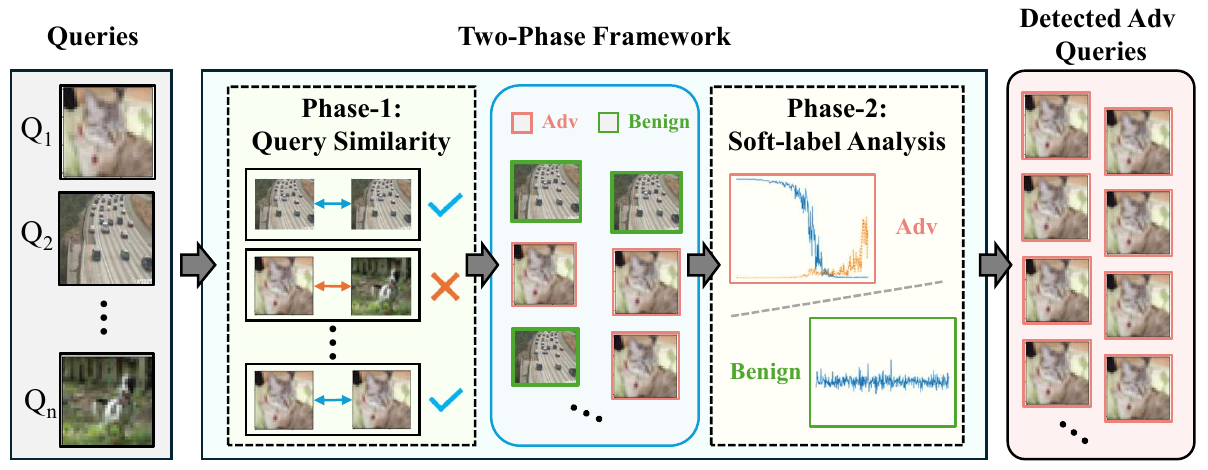}
    \caption{Our proposed two-phase detection framework, which identifies adversarial queries
    within a larger query sequence through a two-phase process. }
    \label{fig:intro_method}
\end{figure*}

Deep neural networks (DNNs) have been shown to be highly susceptible to adversarial attacks, which are carefully crafted imperceptible perturbations
that cause incorrect classification \cite{szegedy2014intriguingpropertiesneuralnetworks}. 
  Even in  the restrictive black-box setting,  these attacks  remain  effective, though they typically require large number of quires to probe the model's behavior. 
  This reliance on frequent queries leads to a defense approach, 
known as Stateful Detection (SD).
SD identifies adversarial queries based on their similarity.
For instance, Chen et. al. \cite{chen2020stateful} measures similarity based on $\ell_2$ distance in a lower dimensional encoded space, whereas 
Blacklight \cite{281294} and PIHA \cite{choi2023piha} evaluate 
query similarity via hashing. Once adversarial attack queries
are detected, the defender can take follow-up actions such as banning the accounts involved,
reporting such instances to authorities, 
performing forensic investigations on these queries, etc.  

Various black-box adversarial attacks have been developed to evade SD methods that are based on query similarity. For example, 
Oracle-guided Adaptive Rejection Sampling (OARS) \cite{Feng_2023} demonstrates
that it is possible to add random perturbations to adversarial attack queries to evade detection. 
This poses a challenge to SDs which
rely on query similarity (e.g. Blacklight) as a much lower similarity threshold is required, which
incurs a higher false positive rate. 
Furthermore, black-box adversarial attacks on features in the frequency domain 
are empirically shown to bypass Blacklight detection \cite{li2022decision}.

Many SDs  measure similarity in a reduced dimensional space or by fast comparison of  hashes, which  could be viewed as efficient approximations of some implicit computationally expensive distance function. To illustrate, performance of the hashing methods are often empirically evaluated  on various geometrically transformed images (e.g. translation, cropping, etc), suggesting that the intended similarity is some implicit function  that  maximizes over the space of possible geometric transformations. The implicit  similarity is computationally expensive  and  thus it is crucial to have an efficient approximation, for instance, through matches of  hashes.
We highlight that the efficient approximations, if not robust, are susceptible to ``adversarial attacks'' in the sense
that an adversary may construct queries that are similar in the implicit metric, but are approximated
as non-similar. We illustrate this with a lightweight ``adversarial attack''
on Blacklight. By incorporating  with the Boundary Attack \cite{brendel2017decision},
this attack could derive 
adversarial examples with small distortion  while evading Blacklight detection. Compared to OARS \cite{Feng_2023}, this lightweight attack does not require a large number of probes on the SD.
    In response to the above,   we formulate the underlying distance function,   and  propose employing probabilistic hashing to approximate the distance function, while hiding the randomness  from the adversary to mitigate adversary attack.

Beside similarity in the query space, we observe that
adversarial query sequence exhibits some form of temporal correlation among  the query results. For instance, a clean surveillance video would likely contain many similar frames, which could be wrongly detected as attack. However,  the soft labels of these clean frames are expected to be independent to each other temporally  which is not the case for queries generated by the black-box attack. This is because the attack adaptively generate the query based on the decision which relates to the soft-labels, and thus would generate queries that depend on early queries. Such temporal correlation could be used to lower the false detection rate.


With these observations, we derive a two-phase detection
framework, as illustrated in Figure \ref{fig:intro_method}. Given a sequence of queries, 
the first phase flags out subsequences of similar queries using the probabilistic hashing. These subsequences will
be passed on to the second phase, where the soft-label behavior of these queries
will be analyzed for any temporal behavior. 


We evaluate the efficacy of the proposed
framework with existing state of the art query-based black-box attacks, as well as
adaptive 
attacks that are specifically designed to evade detection 
(e.g. OARS, and the attack described in \cite{li2022decision}).
The experiments shows that the framework
successfully detects queries from such attacks with high accuracy, and is robust against adaptive probing
using OARS.


Our contributions are as follows, 
\begin{enumerate}
    \item We observe 
    that soft-labels of adversarial attack query sequences exhibits temporal behavior,
    which can be used to distinguish between benign and adversarial query sequences. 
    \item 
    The notion of  ``adversarial attacks''
    on similarity metrics is introduced, and we demonstrate
    such a lightweight attack on the similarity metric employed by Blacklight. 
    \item This paper proposes a similarity function and an approximation (salted randomized quantization) that is robust.
    \item  We propose a two-phase framework to identify adversarial query sequences within a larger query sequence by leveraging both query similarity and the temporal correlation of soft-labels. Current stateful detection (SD) methods that rely solely on query similarity often require a high threshold for the similarity function to minimize the false positive rate. However, this approach leaves them vulnerable to adaptive probing attacks, such as OARS. By incorporating the temporal correlation of soft-labels, we can adopt a lower threshold for feature similarity, thereby mitigating the risk of adaptive attacks. The soft-label is used to filter out the increased false positives incured. 
    \item The performance of the proposed method is empirically evaluated. 
    In particular, 
    on sequences with both benign and malicious queries, 
    we 
     achieve a high true positive rate (1.00, v.s. 0.97 in other SOTA SD) when identifying adversarial query sequences, 
    and a low false positive rate (0.06, v.s. 0.42 in other SOTA SD). Furthermore, our proposed SD incurs modest time and memory overhead, making it a practical and scalable defense. 
\end{enumerate}
This paper is organized as follows: Section \ref{section:background} gives an overview of black-box adversarial
attacks and SDs. Section \ref{section:threat_model} presents the threat model and defense's goal. 
Some properties of adversarial query sequences are highlighted in Section \ref{section:query_property}, 
and a detection framework is proposed in Section
\ref{section:proposed_framework}.
We evaluate our proposed framework in Section \ref{section:experiment}.

\section{Background} \label{section:background}
This section gives a brief overview and discusses prior work on black-box adversarial attacks and
some existing defense methods.

\subsection{Black-box Adversarial Attacks}
Most state of the art deep neural network (DNN) classifers are susceptible to adversarial attacks, 
in the form small imperceptible perturbations added to inputs which result in image misclassification \cite{szegedy2014intriguingpropertiesneuralnetworks}. 
Adversarial attacks can be broadly classified as white or black-box attacks. In the white box setting, 
the adversary has access to detailed model information, including the model architecture and weights, 
whereas in the black-box setting, the adversary only has query access to the model (for example, 
through an API) and receives either the class label (i.e. hard-labels) or the classification
probabilities (i.e. soft-labels). 
There are three main types of black box attacks, a brief explanation of each type is given below. 
\begin{enumerate}
    \item \textbf{Transfer Based Attacks}. The adversary will attempt to train a surrogate model
    which will approximate the behavior of the target model. White-box adversarial attacks on
    the surrogate model have been empirically shown to be transferable to the target model, even
    when the exact architecture of the target model is unknown. The drawbacks of this method is that
    (i) the training data used to train the surrogate model must closely resemble the training data
    of the target model, and (ii) it is computationally expensive to construct surrogate models for target models trained on complicated models (e.g. models trained on ImageNet).
    \item \textbf{Score Based Attacks}. The adversary approximates the gradients using the 
    classification probabilities outputted by the model. This method has been empirically shown to 
    provide good gradient estimates, which can be used to perform white-box attacks. However, 
    a defender can easily defend against this by either truncating or adding noise to the output
    classification probabilities. 
    \item \textbf{Decision Based Attacks}. The adversary approximates the decision boundaries
    using only the classification outputted by the model. This attack requires the least
    amount of information, and is the most realistic threat model. The drawback of
    decision based attacks are that they require a large number of highly similar queries, often
    up to tens of thousands. 
\end{enumerate}
In this study, we will focus on both decision and score
decision based black-box attacks as they represent the most
realistic attack setting. 

\subsection{Existing Defenses}
There are several 
existing and recent defences against decision and score based 
black-box attacks. They rely 
on the observation that these attacks require
an attacker to probe the model with similar queries. We will
briefly describe some of these defenses.
\begin{enumerate}
 \item \textbf{Account Based Query Similarity Detection \cite{chen2020stateful}} 
    This paper introduces the notion of
    stateful detection, where information associated with the query state (e.g. account ID) 
    can be used together with query similarity to identify adversarial query sequences. A query 
    encoder is used to reduce queries into a latent space representation, where similarity checking is performed. 
    In its original implementation, it can be bypassed if the adversary has access to multiple accounts. 

\item \textbf{Blacklight \cite{281294}} 
Blacklight operates by generating a set of hash fingerprints for each query. 
This set is constructed by applying a hash function to a sliding window of fixed size across the query. 
This hash fingerprint set is then compared against
a global database of past hash fingerprints, and flagged as malicious if
there is a significant overlap. For efficiency, the authors empirically
show that it suffices to use a sliding window with no overlap, 
and compare only a randomly chosen subset of the hash
fingerprint set with the global database. However, there are known
decision based black-box attacks which have been shown to bypass
Blacklight, for example \cite{li2021f}. Furthermore, we managed to evade
Blacklight through a modified boundary attack \cite{brendel2017decision}, details 
of this attack on Blacklight is presented in Section \ref{subsubsection:adv_attack_similarity}.

\item \textbf{Boundary Detection \cite{li2022decision}} Boundary detection detects
adversarial queries through a key insight that adversarial examples are likely
located close to the decision boundary of the classifier. A query, $x$,
is deemed to
be close to the decision boundary if $p_1(x) \approx p_2(x)$, where $p_1$ and
$p_2$ are the top-1 and top-2 class probabilities, respectively. Therefore, 
a sequence of queries $x_1, \dots, x_n$ is likely adversarial if
$p_1(x_k) \approx p_2(x_k)$ for a significant number of $k\in\{1,\dots,n\}$.
Boundary detection works at the account level, by flagging out accounts
with a large number of such queries. 
\end{enumerate}
Both Stateful Detection and Blacklight rely
on feature similarity, which means that a series of similar benign queries (such 
as stills from a CCTV) may flagged as malicious. 
Account level defenses can be bypassed if an attacker has access to multiple accounts.
In Boundary Detection, if an adversary manages to submit queries near the decision boundary, they are likely close to creating
a successful adversarial example.

\section{Threat Model and Problem Setup} \label{section:threat_model}

The central problem we are interested in tackling is identifying queries arising from adversarial attacks, 
amidst an overall query sequence consisting of both legitimate and adversarial queries. 

\subsection{Threat Model}
The adversary has black-box access to a machine learning model, on which the adversary can adaptively query
and receive the model output.   
The eventual goal of the attacker is to construct an adversarial example by adaptively querying the model. 
Given a clean input $x$, the attacker wants to find 
an adversarial example $\Tilde{x}$, where $x$ and $\Tilde{x}$ are classified differently, but $\Tilde{x}$ is similar to  $x$ w.r.t. some distance metric, i.e. 
$d (x ,\Tilde{x} ) < \epsilon $ for some distance metric $d$ and fixed $\epsilon>0$.
We assume that the attacker employs a query-based black-box adversarial attack algorithm
to construct the adversarial example. 

To obfuscate his adversarial intent, the adversary may hide the query sequence 
in another sequence of benign queries. In the case of the model being queried
through an online API, besides submitting both benign and malicious queries, 
the attacker could also submit his queries through
multiple accounts to further evade detection. Hence, we also consider scenarios where a smaller subset of the  queries is present in the sequence.

\subsection{Defender's Goal}
The defender has total control over the model, including the full model 
classification outputs
and queries received. Given a sequence of queries received,
the defender's goal is to identify query subsequences that constitute part of
an adversarial attack. 
The effectiveness of the defense can be measured
by:
\begin{enumerate}
\item Given a sequence, determine whether it consists of an attack (true positives vs false positive). The defender does not know the proportion of attack queries within the sequence, if any.  
\item Given a sequence that is known to contain an attack, identify the queries that are the attack.   Similarly, the defender does not know the proportion of attack queries. 
\end{enumerate}



\section{Properties of Queries} \label{section:query_property}

Under the threat model, the defender does not know which black-box attack is employed by the adversary. Nonetheless, most known black-box attacks employ some efficient searching strategy
that adaptively queries a model while locating the decision boundaries. This leads to the following
two characteristics in the resulting query sequence, that is (i) the queries are similar,
and (ii) the soft-labels of the query sequence exhibits temporal correlation.



\subsection{Query Similarity} \label{subsection:similarity_detection_mechanism}



 For a given SD, we consider two distance measurement functions: 
 the intended similarity metric, and its corresponding similarity approximation function.

 \subsubsection*{Motivation}
 
Known black-box attacks typically generate ``similar'' queries to explore the decision boundary. While there does not exist a universal
query similarity metric among such attacks, we observe that many attacks (e.g. \cite{brendel2017decision,guo2019simple,andriushchenko2020square}), either implicitly or explicitly,  employ normalized $\ell_2$ norm, which
results in queries similar in $\ell_2$. 

Although the normalized $\ell_2$ norm is commonly adopted by the attacks, 
most SDs do not directly employ the same metric. 
This is likely due to the following concerns.
\begin{enumerate}
    \item \textbf{Security against evasion.} An adversary may attempt to evade traditional similarity norms by 
    retaining important region while changing the rest of the query. For example, when the queries are images, the attacker
    may keep and/or translate the sub-image corresponding to the region of concern, while varying the background, hoping that the results transfer across
    the queries. Therefore, a distance function that measures distance between ``key regions'' across the queries could detect such queries. 
    \item \textbf{Computational efficiency.} It is computationally expensive to compute the exact distance metric, and thus some
    efficient approximation is employed (e.g. Blacklight employs a randomized hash on the quantized queries, and PIHA employs a perceptual hash). 
\end{enumerate}

 \subsubsection*{Intended Similarity Function} \label{subsubsection:intended_similarity}
Let us consider the following similarity metric in Definition \ref{dfn:intended}. 
\begin{definition}[\emph{Intended Similarity Function}] \label{dfn:intended}
    Given two images
${\bf x}$  and ${\bf y} $, let $S_{\bf x}$ and $S_{\bf y}$ be the set of all subimages (of some pre-deteremined dimension and size $w$) in $\bf x$ and ${\bf y}$ respectively,  where   a subimage of $\bf x$ is the image within a smaller square window in ${\bf x}$. Let us define
\[
    D( {\bf x}, {\bf y}) = \min_{x' \in S_{\bf x}, y' \in S_{\bf y}} d(x', y'),
\]
where $d$ is some predetermined similarity metric, for example, 
$\ell_2$ distance.
\end{definition}

The intended similarity function is the 
smallest  $\ell_2$ distance between pairs of subimages from ${\bf x}$ and ${\bf y}$ of the predetermined size and dimensions (For example, in Blacklight, this refers to rectangles that have a length of one unit and a specific, predetermined width). 
$D$ will be the intended similarity metric used in our two-phase framework.

\subsubsection*{Approximated  Similarity Function}
In practice, computing the intended similarity metric is impractical due to the high computational cost of known subimage-matching algorithms \cite{DBLP:journals/corr/abs-2408-16445}. Consequently, existing stateful detection algorithms (e.g. Blacklight and PIHA) rely on efficient approximations to estimate this similarity metric, stated in Definition \ref{dfn:approx}. 

\begin{definition}[\emph{Approximated Similarity Function}] \label{dfn:approx}
    Given positive integers $r,w$ and $k$, let ${\bf x}$ and ${\bf y}$ be two images, and let $S_{\bf x}$ and $S_{\bf y}$ denote the sets of all subimages of pre-determined dimensions and size $w$ extracted from ${\bf x}$ and ${\bf y}$, respectively. For any two distinct subimages $x, x' \in S_{\bf x}$, the maximum overlap between $x$ and $x'$ is bounded $k$, where $k < w$. 
    Let $R_{\bf x} \subseteq S_{\bf x}$ and $R_{\bf y} \subseteq S_{\bf y}$ represent subsets of $r$ uniformly sampled subimages from $S_{\bf x}$ and $S_{\bf y}$, respectively. The \emph{approximated similarity function}, denoted $\hat{D}_{r,w,k}({\bf x}, {\bf y})$, is defined as:
    \[
        \hat{D}_{r,w,k}({\bf x}, {\bf y}) = \min_{x' \in R_{\bf x}, y' \in R_{\bf y}} d(x', y'),
    \]
    where $d(\cdot, \cdot)$ is the similarity metric used in the intended similarity function.
\end{definition}
Since $\hat{D}_{r,w,k}({\bf x}, {\bf y})$ is a probabilistic function, we show that on expectation, under mild assumptions, $\hat{D}_{r,w,k}({\bf x}, {\bf y})$ closely approximates the intended similarity function $D( {\bf x}, {\bf y})$. In Theorem \ref{thm:similarity_metric_bound}, we provide the relationship between the expectation of the probabilistic approximation function $\IE[\hat{D}_{r,w,k}({\bf x}, {\bf y})]$ and the intended similarity function  $D( {\bf x}, {\bf y})$. In both cases, the expected bias between $\IE[\hat{D}_{r,w,k}({\bf x}, {\bf y})]$ and $D( {\bf x}, {\bf y})$ converges to 0, when more samples are used to approximate $\hat{D}_{r,w,k}({\bf x}, {\bf y})$.
\begin{theorem}\label{thm:similarity_metric_bound}
    Let ${\bf x}$ and ${\bf y}$ be two images of $n$ pixels each, where each pixel value lies in the interval $[0, q]$ for some $q > 0$ (e.g., 255). 
    Then, the expectation of the approximate similarity function can be bounded in terms of the intended similarity function as follows:
    Let $N = \lfloor \frac{n-w}{k}\rfloor$ be the number of windows in $S_{\bf x}$ and $S_{\bf y}$. 
    \[
    \begin{split}
        D({\bf x}, {\bf y}) &\leq \IE[\hat{D}_{r,w,k}({\bf x}, {\bf y})] \\ 
        &\leq \left(1- \left( \frac{N}{N-r} \right)^{-\alpha }\right)D({\bf x}, {\bf y})   + \left(\frac{N}{N-r} \right) ^{-\alpha} \IE[d(W_1, W_2)],
    \end{split}
    \]
    where $\alpha = |\{ x' \in S_{\bf x}, y' \in S_{\bf y} : d(x', y') = D({\bf x}, {\bf y}) \}|$ denotes the number of pairs of subimages in $S_{\bf x}$ and $S_{\bf y}$ that achieve the minimum distance across all possible pairs of subimages, and $W_1, W_2 \stackrel{iid}{\sim} Unif([0, q])^w$
\end{theorem}

\begin{proof}
    We first show that $D( {\bf x}, {\bf y}) \leq \IE[\hat{D}_{r,w,k}({\bf x}, {\bf y})] $. By the definition of $D( {\bf x}, {\bf y}) $, for all $x \in S_{\bf x}, y \in S_{\bf y}$ such that  
    $d(x', y') \geq D( {\bf x}, {\bf y}) $. Therefore,  $D( {\bf x}, {\bf y}) \leq \IE[\hat{D}_{r,w,k}({\bf x}, {\bf y})] $.

    Next, we prove the upper bound.
    The probability that none of the $\alpha$ windows are sampled is 
    \[
    \begin{split}
        p & := \prod_{i=0}^{r-1} 1 - \frac{\alpha}{N-i} \\
        & \approx \exp\left(- \sum_{i=0}^r \frac{\alpha}{N-i}\right) \\
        & \approx \exp\left(-\alpha \log \frac{N}{N-r} \right)\\
        & =\left(\frac{N}{N-r} \right) ^{-\alpha }
    \end{split}
    \]
    where we made use of the fact that $1-x \approx \exp(-x)$ for $0<x<1$, and $\sum_{i=0}^r \frac{1}{N}-i \approx \log N - \log (N-r) = \log \frac{N}{N-r}$. Under this scenario, the expected distance between two randomly sampled windows is $\IE[d(W_1, W_2)]$.
    On the other hand, the probability that at least one of the $\alpha$ most similar windows is selected is $1-\left(\frac{N}{N-r} \right) ^{-\alpha }$, and the output of $\hat{D}_{r,w,k}({\bf x}, {\bf y})$ is $D({\bf x}, {\bf y})$.
    Therefore, the upper bound follows from taking a weighted average of the two possible outcomes based on their respective probabilities. 
\end{proof}

The approximate similarity function $\hat{D}_{r,w,k}({\bf x}, {\bf y})$ is inherently probabilistic in nature. The expected bias between the intended similarity metric $D({\bf x}, {\bf y})$ and its approximate counterpart decreases as the number of sampled windows $r$ increases, and it also depends on the overall similarity between ${\bf x}$ and ${\bf y}$, quantified by the parameter $\alpha$. However, in practical scenarios, $r$ is typically much smaller than the total number of possible windows $N$ to ensure computational efficiency. An adversary could exploit this limitation by introducing small distortions to either ${\bf x}$ or ${\bf y}$, which would reduce the value of $\alpha$. This reduction in $\alpha$ would, in turn, increase the expected bias between the probabilistic approximate similarity function and the actual similarity function. In Definition \ref{dfn:adv_attack_similarity}, we formally define and quantify the strength of such ``adversarial attacks'' on the approximated similarity functions

\begin{definition}[\emph{$(\epsilon, \delta)$-Adversarial Example on Metric}] 
\label{dfn:adv_attack_similarity}
    Let $D$ be the actual similarity metric, and $\tilde{D}$ be its corresponding approximated similarity function. We say that $( {\bf x}, {\bf x}')$ is an $(\epsilon, \delta)$-adversarial example on ($D$, $\tilde{D}$) if $D( {\bf x}, {\bf x}') \le \epsilon$ but
    $\widetilde{D}({\bf x}, {\bf x}') \ge \delta$.
\end{definition}
Given a pair of queries, the notion of $(\epsilon, \delta)$-adversarial examples allows us to quantify the 
discrepancy between their distances under the intended metric $D$ and approximated similarity function $\widetilde{D}$. 
For an approximated distance function to be robust, it should be difficult to construct
an $(\epsilon, \delta)$-adversarial example on $(d, \tilde{d})$ with a large gap $\delta - \epsilon$.

\subsubsection*{Adversarial Attack on Similarity Function} \label{subsubsection:adv_attack_similarity}
We propose a lightweight ``adversarial attack'' on the similarity detection
metric used by Blacklight, which does not require significant probing on the SD. In Blacklight, 
incoming queries are quantized (with quantization constant $q$) and partitioned
into overlapping moving windows of size $w$, and a few windows are selected 
randomly. They are then compared with a database of previously sampled windows. 

This lightweight attack is performed on image queries,
by randomly choosing a pixel in every non-overlapping window of size $w$, and adding
$c \times q$ to that pixel, where $c$ is randomly chosen from $\{-1, 1\}$. 
In doing so, we ensure that with high probability, any window of length $w$ 
sampled will not overlap with windows sampled from previous queries. 
In the case of Blacklight, it is possible to estimate
$q$ and $w$ using trial and error. 
We note that using an offset of $q/2$ works in approximately $65\%$ of the time. 

Furthermore, we modified the existing Boundary attack \cite{brendel2017decision}
to apply this random mask when querying the model, which successfully evades Blacklight
using the recommended configurations and randomized salt.
The adversarial example generated is in Figure \ref{fig:bl_mask_attack}.
The presence of such ``adversarial attacks''  underscores
the importance of choosing a robust similarity metric. 

\begin{figure}
\centering
  \includegraphics[width=.3\linewidth]{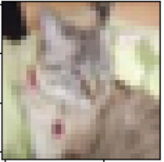}
       \includegraphics[width=0.3\linewidth]{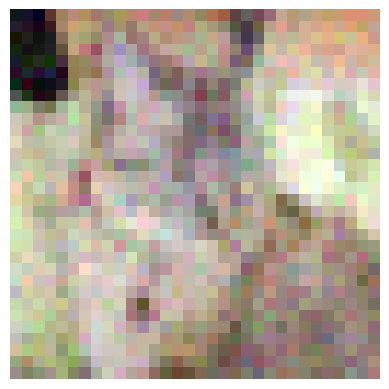}
      \includegraphics[width=.3\linewidth]{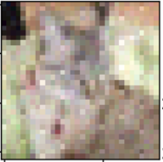}
    \caption{From left to right, the original image, adversarial example generated using OARS, 
    and adversarial example generated using our lightweight attack.
    }
    \label{fig:bl_mask_attack}
\end{figure}

\subsection{Temporal Correlation of Soft-Labels}


 For a sequence of benign queries, 
the soft-labels are not expected to exhibit any temporal trends, whereas the soft-label sequence
corresponding to an adversarial query sequence is likely to display a clear shift over time. 
See Figure \ref{fig:shifting labels} for an example of this
described temporal soft-label behavior.

\begin{figure*}
        \centering
        \includegraphics[width=0.45\linewidth]{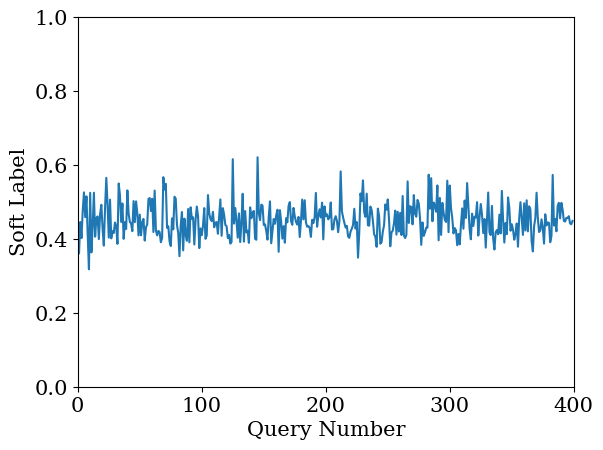}
            \includegraphics[width=0.45\linewidth]{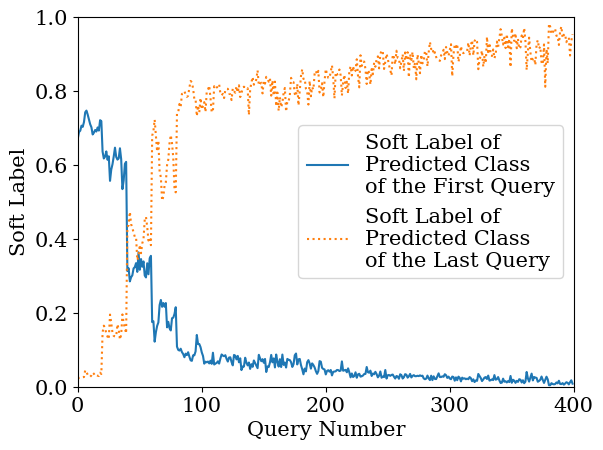}
        \caption{(From left to right) Soft-label behavior of a sequence of similar queries (stills from CCTV feed)
        vs queries generated from adversarial attack. For the camera image sequence, the soft-label
        corresponds to the predicted class of the first image in the sequence. A smaller numbered query arrives earlier.}
        \label{fig:shifting labels}

\end{figure*}
To determine if the soft-label sequence exhibits any temporal trends, 
we extract the soft-labels associated with the predicted class of the first query, for the entire query sequence. 
The Ljung–Box test is applied to this soft-label sequence
    to test for serial correlation
    \cite{box1970distribution}. If the $p$-value from this test falls below a predetermined 
    threshold, the sequence is deemed to exhibit serial correlation and hence
    is likely to be temporal \cite{matyas1996serial}, and is flagged as adversarial.
    This behavior corresponds with the objectives of an adversarial 
attack, which is to identify and iteratively refine perturbations which will result in misclassification.
Therefore, a sequence that is highly similar, and exhibits temporal behavior, is likely to constitute
an adversarial attack sequence.

\section{Proposed Detection Framework} \label{section:proposed_framework}

In this section, we use the observations described earlier to construct our proposed
detection framework for adversarial attack queries. It 
is a two phase process, which given a sequence consisting of clean and adversarial queries, 
(i) detects subsequences of queries that are similar, and (ii) passes this subsequence
for soft-label analysis to determine if it is adversarial.

\subsection{Phase 1: Detecting Subsequences of Similar Queries} \label{subsection:phrase1}

Given an input query  ${\bf{x}}=(x_1, \dots, x_n)$, where $x_i \in [0, N]$ for all $i=1, \dots, n$,
for an integer $N>0$, 
and a fixed randomly chosen salt $s$,
quantization constant $q$, window size $w$ and maximum overlap $k$, 
\begin{enumerate}
    \item Perform \textit{salted randomized quantization} on ${\bf{x}}$, as follows:  
Set each entry $x_i$ as $\lfloor (x_i + s) \mod N /q \rfloor + B_i$, where $B_i \sim \text{Bernoulli}((1/w) (x_i\mod q)/q)$,
where $s$ is a salt chosen randomly over $[0,1]$ uniformly. 
    \item Randomly select $r$ windows of size $w$ from the quantized query, where the pairwise overlap
    between windows is at most $k$.
    \item Flag ${\bf{x}}$ as similar (w.r.t. some previous query) if the number of windows overlap with previously encountered ones exceeds a 
    predetermined threshold. 
    \item A list of candidate adversarial subsequences is maintained. To determine the subsequence membership
    for a flagged query, we perform
    the above comparison method above on each individual subsequence.
    \item Subsequences that are updated and exceed a predetermined length threshold
    are sent to Phase 2 for soft label analysis.
\end{enumerate}
The constants $q, w$ and $k$ are hyperparameters, to be tuned in Section \ref{subsection:det_threshold}.

\subsubsection*{Salted randomized quantization} \label{subsubsection:srq}
{Salted randomized quantization} introduces randomness into the quantization process, by introducing
a random salt and randomly rounding values up. This improves the robustness of the proposed similarity
metric against reverse engineering 
and adaptive attacks. Our emperical evaluations (Section \ref{subsubsection:adv_attack_similarity}) show that
our framework using {salted randomized quantization} achieves comparable detection accuracy compared to using
normal quantization, and is more robust against adaptive attacks (e.g. OARS). 
\paragraph{Robustness Against Adaptive Probing}
An adversary may adaptively probe the SD to determine the minimum perturbation required to bypass the similarity function. Specifically, this involves identifying the $(\epsilon, \delta)$ "adversarial attack" on the similarity function (as defined in Definition \ref{dfn:adv_attack_similarity}) that is sufficient to evade detection. 
Salted randomized quantization significantly increases the difficulty for an adversary to find a perturbation magnitude that avoids detection without degrading the quality of the adversarial attack. In Section \ref{subsection:oars_expt}, we empirically demonstrate that an adversary must introduce a substantial amount of perturbation to evade detection, ultimately failing to produce a successful adversarial example.


\subsubsection*{Improving computational efficiency}
The following measures may be taken to improve computational efficiency:
\begin{enumerate}
    \item The database of sampled windows is stored in a Bloom filter,
    which facilitates efficient insertion and retrieval.
    \item To reduce computation time and overhead, we limit the following:
    \begin{enumerate}
        \item The number of subsequences maintained. 
    The subsequences are stored in a queue data structure, ordered by the time since the last update. 
    If the queue length exceeds a certain threshold, subsequences that were not recently updated
    will be evicted from the queue first. 
        \item The number of queries in each subsequence. If the length of a subsequence exceeds
        a certain threshold, the oldest query will be removed. 
        \item Matching an individual query to existing subsequences can be 
        done concurrently through multiprocessing. 
    \end{enumerate} 
\end{enumerate}

\subsection{Phase 2: Analysing Soft-labels of Query Subsequences} \label{subsection:phrase2}

Subsequences of similar queries identified in Phase 1 will be sent to Phase 2 for soft-label 
analysis. 
Given a  sequence of queries ${\bf{x}}_1, \dots, 
{\bf{x}}_n$, a classifier which outputs the predicted
hard label as $C(\cdot)$ and soft-label of class $c$ as $p_c(\cdot)$, we propose the following to
analyze soft-label for temporal trends.
\begin{enumerate}
    \item Let $c_1 = C({\bf{x}}_1)$, and obtain the sequence of soft-labels $p_{c_1} := p_{c_1}({\bf{x}}_1), \dots, p_{c_1}({\bf{x}}_n)$
    from the model output.
    \item 
     Apply the Ljung–Box test on the sequence $p_{c_1}$
    to test for serial correlation (and hence, temporal behavior) in the soft-label sequence.
    If the $p$-value from this test is below a predetermined 
    threshold, flag the sequence as adversarial. This threshold is a hyperparameter, 
    which is tuned in Section \ref{subsection:threshold_pv}.
\end{enumerate}
The analysis of soft-labels will prevent wrongly flagging similar but benign query sequences
(e.g. stills from camera recordings or images with a similar backgrounds) from being wrongly
flagged as malicious. 

\subsubsection*{Choice of Test for Temporal Correlation}
The Ljung-Box test is a statistical test that admits the null hypothesis that a given sequence of soft-labels does not exhibit spatial temporal behavior. 
Besides the Ljung-Box test, we also evaluated two other commonly used tests in practice: (1) the Lagrange Multiplier test and (2) the Breusch–Godfrey test. Our choice of the Ljung-Box test is motivated by the following reasons:

\begin{enumerate}
\item \textbf{Flexibility in lag selection.} The Ljung-Box test evaluates temporal correlation across multiple lags simultaneously, providing a comprehensive assessment of autocorrelation.

\item \textbf{Empirical performance.} During our empirical evaluations, the Ljung-Box test consistently achieved the best performance across all datasets and models. This finding aligns with prior studies (e.g., \cite{HASSANI201981, dare2022comparison}), which demonstrate that the Ljung-Box test is a robust statistical method for detecting temporal correlation in real-world datasets.
\end{enumerate}
We emphasize that our two-phase framework is designed to be flexible and can seamlessly accommodate alternative tests for temporal correlation if desired.

\subsection{Combined Two-phase Method}
As highlighted earlier, the goal of Phase 1 is a computationally efficient way to
identify subsequences of similar queries, while the goal of Phase 2 is to 
determine if a given subsequence consitutes an adversarial attack based on their 
soft-label behavior. 
To improve detection coverage, we can reduce the detection thresholds in 
    Phase 1 to reduce the false negative rate, at the expense of a higher false positive
    rate. These false positives can then be filtered
    out in Phase 2.

\section{Experiments} \label{section:experiment}
Our proposed methods are empirically evaluated using five exisiting black box adversarial attacks, which encompass both query and score based attacks. The aim of our experiments are,
\begin{enumerate}
    \item Determine the parameters and detection thresholds, and evaluate whether these
    parameters and thresholds are generalizable to previously unseen adversarial attack queries;
    \item Evaluate the performance of the proposed framework against several SOTA adversarial attack algorithms;  and  
    \item Evaluate the performance of the proposed framework against adaptive attacks designed to 
    evade SD detection. 
\end{enumerate}
We measure accuracy using (i) true positive rates (i.e. proportion 
malicious queries that are detected), and (ii) false positive rates (i.e. benign queries that
are wrongly flagged as adversarial).

\subsection{Experimental Setup} \label{subsection:expt_steup}

We evaluate our proposed method using the models and datasets given in Table
\ref{table:model_data}. The ResNet50 model for ImageNet was downloaded
from PyTorch, and the CNN model for CIFAR10 was downloaded from \cite{art2018}.


\begin{table}
\caption{Detailed information of models and datasets.}
\centering
\begin{tabular}{c|cccc}
\toprule
Dataset  & Classes & Model     & Input Shape               & \begin{tabular}[c]{@{}l@{}}Top-1 \\ Accuracy\end{tabular} \\ \midrule
ImageNet & 100     & ResNet50 & $224 \times 224 \times 3$ & 76\%                                                   \\
CIFAR10  & 10      & CNN       & $32 \times 32 \times 3$   & 85\%                                                      \\ \bottomrule
\end{tabular}

\label{table:model_data}
\end{table}

To assess if our proposed method is robust to non-malicious but similar
query sequences, we included the following in the datasets,
\begin{enumerate}
    \item Sequence of images constructed by applying 
    i.i.d $\mathcal{N}(0, 0.1^2)$ noise on a base image.
    The base image is not part of any adversarial query sequence.  \label{img_own_noise}.
    \item Stills taken from the camera feed in
    \cite{kaggle_highway}, examples in Figure \ref{fig:cctv-stills}. \label{img_camera}
\end{enumerate}
The rationale for using (\ref{img_own_noise}) is that an adversary may
attempt to overwhelm query similarity detectors by submitting queries that
are similar but not part of an adversarial attack sequence. 
Including (\ref{img_camera}) allows us to confirm this on real-world 
images, where some non-malicious queries may overlap significantly.

\begin{figure}
    \centering
    \includegraphics[width=0.3\linewidth]{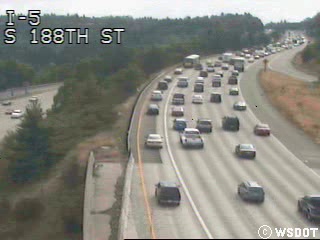}
    \includegraphics[width=0.3\linewidth]{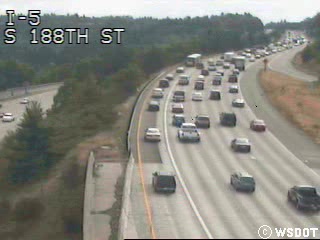}
    \includegraphics[width=0.3\linewidth]{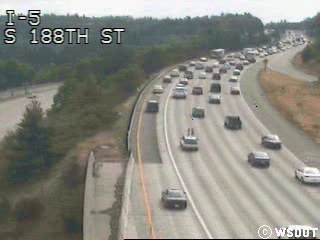}
    \caption{Examples of stills taken from a the camera feeds. There is significant overlap in background features among them.}
    \label{fig:cctv-stills}
\end{figure}

\subsection{Determining Detection Thresholds} \label{subsection:det_threshold}

In this section, we will explain the process and rationale behind the decision thresholds 
chosen for both Phase 1 (query similarity) and
and Phase 2 (soft-label analysis).
These thresholds will be used in the subsequent evaluations. 

\subsubsection*{Parameters in Phase 1} \label{subsection:threshold_query_similarity}

To determine the detection thresholds for Phase 1, we 
performed our proposed query similarity detection method on a sequence of 10,000 queries consisting 
of
\begin{enumerate}
    \item 2,000 images from the camera stills described in Section \ref{subsection:expt_steup};
    \item 5,000 distinct images randomly chosen from the datasets, which do not overlap with any image
    in (3); and
    \item 3,000 images consisting of 10 sequences of 300 adversarial queries generated from
    the Square Attack \cite{andriushchenko2020square}. 
\end{enumerate}
The queries in the sequence are randomly shuffled, but queries in (3) arranged in their original
query ordering. The choice of the Square attack is arbitrary, as we want to assess how well the obtained
threshold generalizes to other adversarial attacks. 

The metric used is the number of windows flagged out as similar in each of the 3 categories. For (1) and (2),
a low false positive rate (FPR) corresponds to fewer windows flagged out, and for (3), a high true positive rate
(TPR) corresponds to a higher number of windows flagged out. 
Since the role of Phase 1 is to flag out suspicious query sequences to undergo soft-label analysis in Phase 2, 
we allow for a higher TPR at the expense of a higher FPR.
The parameters were chosen using a grid search strategy. 
The resulting parameters for each dataset are, 
\begin{itemize}
    \item CIFAR10: $q=80, w=20, r=45, k=15$;
    \item ImageNet: $q=80, w=50, r=75; k=20$.
\end{itemize}
Using the above parameters, we will estimate a decision threshold
which will be used to flag a query sequence as similar, if the number of
overlapping windows exceeds this threshold. The thresholds are,
\begin{itemize}
    \item CIFAR10: 33;
    \item ImageNet: 7.
\end{itemize} 
The distribution of
the number of overlaps by query type, and the corresponding decision
threshold is summarized in Figure \ref{fig:phrase_1_distr_threshold}.

\begin{figure}
    \centering
    \includegraphics[width=0.49\linewidth]{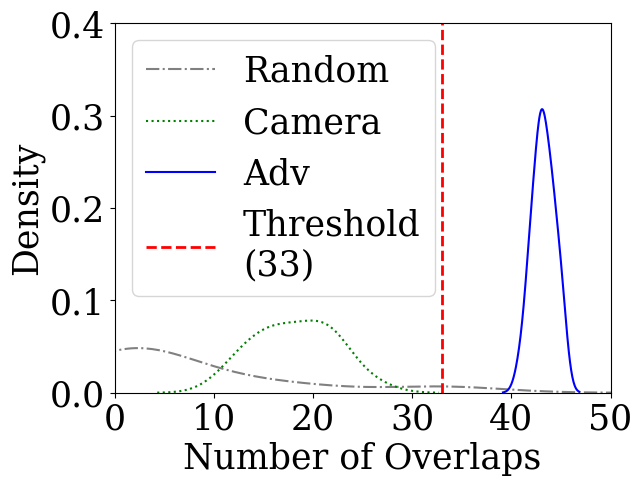}
    \includegraphics[width=0.49\linewidth]{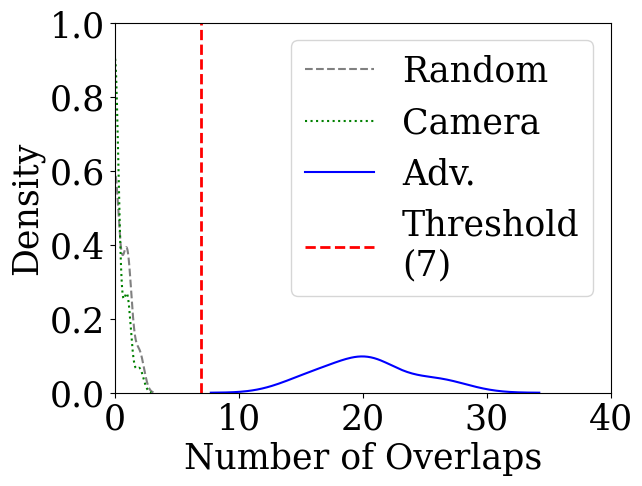}
    \caption{(From left to right) Distribution of the number of overlapping windows (by image category) for the CIFAR10 and
    ImageNet datasets, respectively.}
    \label{fig:phrase_1_distr_threshold}
\end{figure}

\begin{table} [!htbp]
\caption{TPR, FPR and Overall Accuracy for different $p$-value thresholds. Benign queries
include images randomly sampled from the ImageNet and CIFAR10 datasets, 
and adversarial sequences used are of length $15$}
\centering
\begin{tabular}{c|ccc}
\toprule
$p$    & TPR  & FPR  & Accuracy \\ \midrule
0.005 & 0.17 & 0    & 0.58     \\
0.01  & 0.52 & 0    & 0.71     \\
0.015 & 0.73 & 0    & 0.86     \\
0.02  & 0.98 & 0.01 & 0.97     \\
0.025 & 1    & 0.02 & 0.99     \\
0.03  & 1    & 0.02 & 0.99     \\
0.035 & 1    & 0.02 & 0.98     \\
0.04  & 1    & 0.03 & 0.97     \\
0.045 & 1    & 0.05 & 0.95     \\
0.05  & 1    & 0.05 & 0.95     \\ 
\bottomrule
\end{tabular}

\label{table:p_val_threshold}
\end{table}

\subsubsection*{Parameters in Phase 2}  \label{subsection:threshold_pv}
Phase 2 examines the soft-label behavior of a candidate query sequence 
using the statistical test described in Section \ref{subsection:phrase2}. The sequence 
is deemed to exhibit temporal behavior if the $p$-value of the statistical test 
is less than a certain threshold. 
To determine this threshold, 
the following query sequences were constructed.
\begin{enumerate}
    \item Images randomly sampled from the dataset;
    \item Queries generated from an adversarial attack (Square Attack).
\end{enumerate}
The queries are of various lengths to determine the sequence length
needed for successful detection. 
The following metrics are employed, 
\begin{itemize}
    \item False Positive Rate (FPR): Proportion of sequences from
    (1) wrongly flagged as adversarial
    \item True Positive Rate (TPR): Proportion of sequences from
    (2) correctly flagged as adversarial
    \item Accuracy: Proportion of sequences flagged correctly
\end{itemize}
The results are summarized 
in Figure \ref{fig:expt_p_value} where the $p$-values from both CIFAR10
and ImageNet datasets are combined as they are distributed similarly, and are of the same scale 
(which allows for direct comparison). The full results
are available in Table \ref{table:p_val_threshold}.
\begin{figure}
    \centering
    \includegraphics[width=.49\linewidth]{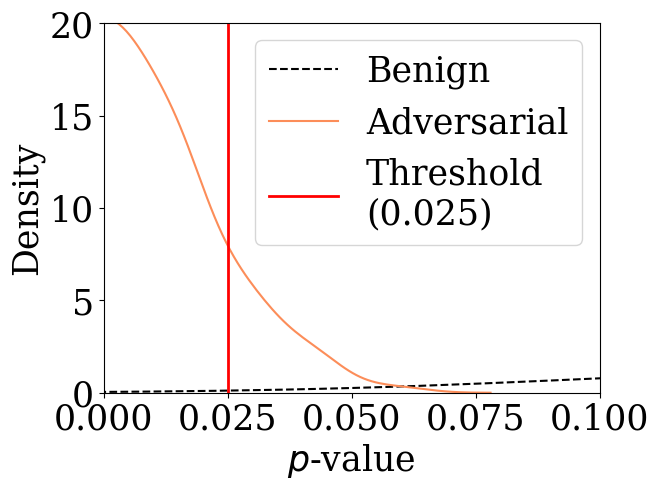}
    \includegraphics[width=.49\linewidth]{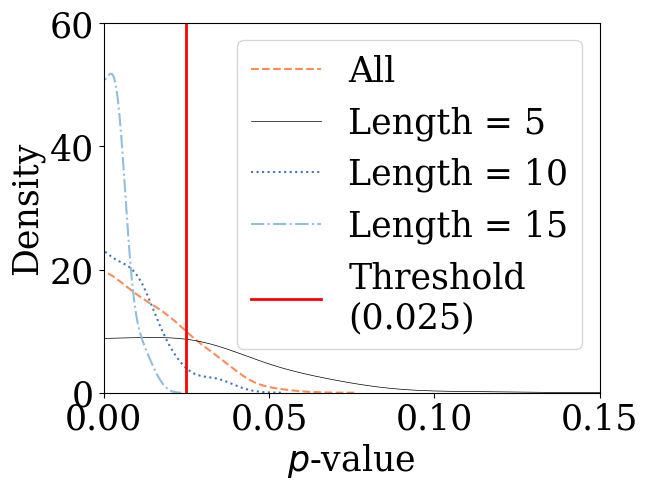}
    \caption{(Left) Distribution of $p$-values obtained for adversarial vs non-adversarial query sequences. 
    (Right) Distrbution of $p$-values obtained for adversarial sequences of different lengths. Sequences with
    $p$-values that are 
    less than the threshold will be flagged as adversarial.}
    \label{fig:expt_p_value}
\end{figure}
Based on Figure  \ref{fig:expt_p_value}, 
\begin{itemize}
    \item As expected, $p$-values of the randomly selected/camera image sequences 
    are on average, larger $p$-values compared to adversarial query sequences. 
    \item The selected threshold of $0.025$ accurately flags out
    all adversarial sequences of length 15. In subsequent evaluations, 
    candidate query sequences flagged in Phase 1 must contain at least 15 queries before proceeding to Phase 2.
\end{itemize}

\subsection{Robustness Against Existing SOTA Query-based Black-box Attacks} 
\label{section:expt_robustness_against_sota}

We evaluate our proposed framework against the following hard and soft-label based attacks,
\begin{enumerate}
    \item (Hard-label) Boundary Attack (BA) \cite{brendel2017decision}
    \item (Hard-label) Hop Skip Jump Attack (HSJA) \cite{chen2019hop}
    \item (Hard-label) HSJA with frequency mixup (f-mixup) \cite{li2022decision},
    which was proposed as an effective method to evade Blacklight detection
    \item (Soft-label) Square Attack (SA) \cite{andriushchenko2020square}
    \item (Soft-label) Simple Black-box Attack (SimBA) \cite{guo2019simple}
\end{enumerate} 
For each attack, 20 image query sequences were generated by first selecting 20 random images from the dataset (ImageNet/CIFAR-10) and applying adversarial attacks. Then, 5,000 images from the respective datasets (not used for the adversarial attack), and 1,000 camera stills from each of two video datasets were randomly selected and shuffled. The first 50 queries from the adversarial queri sequence is inserted into this sequence while maintaining their original order.
The parameters used
for these attacks are those suggested in their respective papers or 
attack implementations. The metrics used for evaluation as follows,
\begin{itemize}
    \item False positive rate (FPR): Proportion of flagged queries that are not part of an adversarial attack. 
    \item True positive rate (TPR): Proportion of flagged queries that are part of an adversarial attack. 
\end{itemize}
We did not report the overall accuracy, as this metric is affected by the proportion of
benign and adversarial queries in the sequence. 
To enable a fair comparison, we
tuned Blacklight's parameters using a grid search strategy, to maximize its accuracy. 
The full results of our evaluation is given in Table 
\ref{table:eval_framework_bl}, together with the performance of Blacklight (using both
original parameters and tuned parameters)
on the same query sequences as a baseline. 

\begin{table*}
\caption{Evaluation results of the our proposed two-phase framework (using salted randomized
quantization), with
respect to Blacklight. Blacklight refers to the default recommended settings,  Blacklight$^{+}$ refers to the method with tuned parameters obtained through a grid search strategy. }
\centering
\begin{tabular}{cc|ccc|ccc|ccc}
\toprule
\multicolumn{2}{c|}{Attack}      & \multicolumn{3}{c|}{BA\cite{brendel2017decision}}                     & \multicolumn{3}{c|}{HSJA\cite{chen2019hop}}                   & \multicolumn{3}{c}{F-mixup(HSJA)\cite{li2022decision}}          \\
\multicolumn{2}{c|}{Defense}     & Blacklight & Blacklight$^{+}$   & \textbf{Ours}      & Blacklight & Blacklight$^{+}$   & \textbf{Ours}      & Blacklight & Blacklight$^{+}$   & \textbf{Ours}      \\\midrule
\multirow{2}{*}{CIFAR10}  & TPR & 1.00       & 0.73          & \textbf{1.00} & 1.00       & 0.68          & \textbf{1.00} & 0.87       & 0.66          & \textbf{1.00} \\
                          & FPR & 0.28       & 0.00          & \textbf{0.00} & 0.21       & 0.00          & \textbf{0.00} & 0.29       & 0.00          & \textbf{0.00} \\\midrule
\multirow{2}{*}{ImageNet} & TPR & 0.83       & 0.42          & \textbf{0.98} & 0.98       & 0.50          & \textbf{1.00} & 0.91       & 0.57          & \textbf{1.00} \\
                          & FPR & 0.49       & {0.00} & \textbf{0.09}         & 0.37       & 0.00 &\textbf{0.04}          & 0.41       & 0.00 & \textbf{0.07}          \\\midrule
                          \midrule
\multicolumn{2}{c|}{Attack}      & \multicolumn{3}{c|}{SA\cite{andriushchenko2020square}}                     & \multicolumn{3}{c|}{SimBA\cite{guo2019simple}}                  & \multicolumn{3}{c}{Average}                \\
\multicolumn{2}{c|}{Defense}     & Blacklight & Blacklight$^{+}$   & \textbf{Ours}      & Blacklight & Blacklight$^{+}$   & \textbf{Ours}      & Blacklight & Blacklight$^{+}$   & \textbf{Ours}      \\\midrule
\multirow{2}{*}{CIFAR10}  & TPR & 1.00       & 0.58          & \textbf{1.00} & 1.00       & 0.68          & \textbf{1.00} & 0.97       & 0.67          & \textbf{1.00} \\
                          & FPR & 0.33       & 0.00          & \textbf{0.00} & 0.47       & 0.00          & \textbf{0.00} & 0.31       & 0.00          & \textbf{0.00} \\\midrule
\multirow{2}{*}{ImageNet} & TPR & 0.94       & 0.61          & \textbf{1.00} & 0.87       & 0.59          & \textbf{1.00} & 0.91       & 0.54          & \textbf{1.00} \\
                          & FPR & 0.37       & 0.00 & \textbf{0.03}          & 0.46       & 0.00 & \textbf{0.07}          & 0.42       & 0.00 & \textbf{0.06}      \\\bottomrule  
\end{tabular}

\label{table:eval_framework_bl}
\end{table*}

\subsubsection*{Analysis}
The experimentation result is summarized in Table \ref{table:eval_framework_bl}.
\begin{enumerate}
    \item \textbf{Performance on Adversarial Queries (TPR)}. Our proposed two-phase detection framework achieves
    comparable performance to Blacklight on adversarial queries. The higher quantization constant used
    improves adversarial query detection coverage.
    \item \textbf{Performance on similar queries (FPR)}. Most of the FPR
    in Blacklight (default setting) are due to similar queries arising from the camera stills dataset. Blacklight was not able to distinguish non-adversarial and adversarial queries that are similar. 
    Our framework suffers from a slightly higher FPR in ImageNet (compared to CIFAR10), due to some images
    having overlapping backgrounds. 
    \item \textbf{Average queries to detection}. On average, adversarial query sequences are identified
    within 20 and 45 queries for the CIFAR10 and ImageNet datasets, respectively. 
\end{enumerate}

\subsection{Detection Ability Against Similar but Benign Queries}

\begin{table*}
 \caption{Evaluation of our proposed framework against Blacklight on adversarial and non-adversarial queries. Blacklight refers to the default recommended settings,  Blacklight$^{+}$ refers to the method with tuned parameters obtained through a grid search strategy.} \label{table:adv_fpr}
\centering\begin{tabular}{c|cccc|ccc|cccc|ccl}
\toprule
\multicolumn{8}{c|}{CIFAR10}                                                  & \multicolumn{7}{c}{ImageNet}             \\ \midrule
\multirow{2}{*}{Defense}  & 
  \multicolumn{4}{c|}{Detected Queries from Each Category} &
  \multirow{2}{*}{TPR} &
  \multirow{2}{*}{FPR} &
  \multirow{2}{*}{ACC}  &
  \multicolumn{4}{c|}{{Detected Queries from Each Category}} &
  \multirow{2}{*}{TPR} &
  \multirow{2}{*}{FPR} & 
  \multirow{2}{*}{ACC} \\ \cmidrule{2-5} \cmidrule{9-12}
                 & I   & II & III & IV  &      &      & \multicolumn{1}{l|}{} & I   & II & III & IV &      &      &      \\ \hline
Blacklight       & 100 & 0  & 97  & 100 & 1.00 & 0.66 & 0.26                 & 100 & 18 & 76  & 96 & 1.00 & 0.63 & 0.28 \\
Blacklight$^{+}$ & 53  & 0  & 0   & 0   & 0.53 & 0.00 & 0.88                 & 41  & 0  & 0   & 0  & 0.41 & 0.00 & 0.85 \\
\textbf{Ours} &
  \textbf{100} &
  \textbf{0} &
  \textbf{0} &
  \textbf{0} &
  \textbf{1.00} &
  \textbf{0.00} &
  \textbf{1.00} &
  \textbf{100} &
  \textbf{0} &
  \textbf{0} &
  \textbf{0} &
  \textbf{1.00} &
  \textbf{0.00} &
  \textbf{1.00} \\ \bottomrule
\end{tabular}

\end{table*}

As observed in Section \ref{section:expt_robustness_against_sota}, assessing query similarity alone
will result in the SD wrongly flagging similar but benign queries as adversarial. 
To assess the performance of our proposed framework,
we tested it using the following query sequences
\begin{enumerate}
    \item Queries generated from each of the 5 adversarial attacks listed in Section \ref{section:expt_robustness_against_sota}.
    \item Randomly chosen images from the CIFAR and ImageNet dataset.
    \item Query sequence constructed by corrupting a base image with Gaussian noise.
    \item Successive images taken from camera stills. 
\end{enumerate}
We created 100 query sequences of length 50 for each of the four sequence types. For the adversarial attack sequences, we generated 20 sequences for each of the five adversarial attack algorithms. The adversarial
attack sequence is constructed by randomly selecting an image from the dataset, and taking the
first 50 queries generated by an adversarial attack on that image. 

The following metric is used to evaluate the performance
of our proposed framework.
\begin{itemize}
    \item Number of sequences flagged as containing adversarial queries for each of the four types of query sequences
    \item True positive rate (TPR): The proportion of sequences from (1) which are flagged as adversarial.
    \item False positive rate (FPR): The proportion of sequences from (2) - (4) which are flagged as adversarial.
    \item Accuracy (ACC): The overall proportion of sequences which are classified correctly. 
\end{itemize}
We compared our framework with Blacklight, where the default version uses its recommended settings
for the respective datasets, while the tuned version uses a grid search approach to optimize
its parameters to improve its overall accuracy. 
The results of our evaluation is presented in Table \ref{table:adv_fpr}. 

We observe the following:
\begin{itemize}
    \item \textbf{TPR}. Our framework achieves a high TPR comparable to Blacklight, which
    demonstrates its detection ability on adversarial queries. 
    \item \textbf{FPR}. Our detection framework achieves a 0 FPR, as sequences of similar
    but benign queries are filtered out in Phase 2. As Blacklight operates using image
    similarity, all sequences that consists of highly similar queries are flagged. 
    \item The tuned version of Blacklight reduces its false positive rate by lowering
    its quantization constant, which results in a lower detection coverage of adversarial query sequences. 
\end{itemize}

\subsection{Robustness Of Phase 1 against OARS} \label{subsection:oars_expt}
Recap that OARS consists of two parts, namely
(1) adaptively probing a SD to determine the magnitude of perturbation needed to evade detection, 
and (2) modifying existing attacks to utilize this magnitude of perturbation. This evaluation serves as an ablation study for the robustness of our proposed query similarity detection in phase 1. 



\subsubsection*{Noise Distribution Needed To Evade Phase 1}
Note that the perturbation needed by OARS to evade phase 1
prevents an adversary from querying a model
meaningfully.
\begin{enumerate}
 \item The quantization constant we applied in Phase 1 is large (80
in our framework v.s. 50 in Blacklight), which implies that any
perturbation to evade our similarity detection phase must be large. 
    \item Randomness introduced in the similarity detection process through
    salted randomized quantization makes it harder for an adversary to 
    bypass similarity detection. 
\item The larger step size (15 (CIFAR10) and 20 (ImageNet), vs 1 in Blacklight (all queries))
means that our framework will cover more portions of the images.
Therefore, any adaptive sampling will require a large amount of noise
across the entire query space.
\end{enumerate}

\subsubsection*{Empirical Evaluation of OARS} \label{subsubsection:oars evaluation}
We proceed to empirically evaluate OARS against our proposed detection framework. 
 Since all the attacks
work in a similar way (i.e. constructing queries using the proposed noise
distribution that is identified via adaptive querying), 
we will evaluate our framework against the OARS-enhanced
boundary attack (hard-label), and HSJA attack (soft-label), with Blacklight (default and tuned versions) as a baseline. 
In our evaluation, we picked 20 images randomly 
from both the CIFAR10 and ImageNet datasets, which were not used in our previous experiments. 
We evaluate the performance of the OARS enhanced attack by 
\begin{enumerate}
    \item The magnitude of adaptive perturbation needed to bypass detection
    \item Attack success rate - An adversarial attack is deemed to be successful
            if it is able to produce an adversarial example within 50,000 queries. 
    \item Number of queries needed to generate an adversarial example. 
    \item Number of queries until detection.
\end{enumerate}

\begin{table*}
\caption{Proposal distribution, attack success rates, and the
number of queries needed for successful adversarial attack on our proposed framework and Blacklight. The numbers in the parenthesis 
refer to the standard deviations} \label{table:evaluate_oars}
\scalebox{0.95}{
\begin{tabular}{c|c|ccc|ccc}
\toprule
\multirow{3}{*}{Assessment}                                                                                 & \multirow{3}{*}{Attack}                                   & \multicolumn{3}{c|}{CIFAR10}                                                                                                                                                                                         & \multicolumn{3}{c}{ImageNet}                                                                                                                                                                                                                     \\ \cline{3-8} 
&                                                           & \multicolumn{2}{c|}{Proposed Framework (Phase 1)}                                                                                                                                     & \multirow{2}{*}{Blacklight} & \multicolumn{2}{c|}{Proposed Framework (Phase 1)}                                                                                                                                     & \multirow{2}{*}{Blacklight}                              \\ \cline{3-4} \cline{6-7}
&                                                           & \multicolumn{1}{c|}{\begin{tabular}[c]{@{}c@{}}Normal \\ Quantization\end{tabular}} & \multicolumn{1}{c|}{\begin{tabular}[c]{@{}c@{}}Salted \\ Randomized\\ Quantization\end{tabular}} &                             & \multicolumn{1}{c|}{\begin{tabular}[c]{@{}c@{}}Normal \\ Quantization\end{tabular}} & \multicolumn{1}{c|}{\begin{tabular}[c]{@{}c@{}}Salted \\ Randomized\\ Quantization\end{tabular}} &                                                          \\ \midrule
\multirow{2}{*}{\begin{tabular}[c]{@{}l@{}}Proposal Distribution \\ (OARS)\end{tabular}}                    & HSJA                                                      & \multicolumn{1}{c|}{$U(-14.9, 14.9)$}                                               & \multicolumn{1}{c|}{$U(-13.7, 13.7)$}                                                            & $U(-6.5, 6.5)$              & \multicolumn{1}{c|}{$U(-20.0, 20.0)$}                                               & \multicolumn{1}{c|}{$U(-17.9, 17.9)$}                                                            & $U(-8.1, 8.1) $                                          \\ \cline{2-8} 
 & \begin{tabular}[c]{@{}l@{}}Boundary\\ Attack\end{tabular} & \multicolumn{1}{c|}{$\mathcal{N}(0, 8.5^2)$}                                                  & \multicolumn{1}{c|}{$\mathcal{N}(0, 8.1^2)$}                                                               & $\mathcal{N}(0, 4.0^2)$               & \multicolumn{1}{c|}{$\mathcal{N}(0, 9.0^2)$}                                                  & \multicolumn{1}{c|}{$\mathcal{N}(0, 8.8^2)$}                                                               & $\mathcal{N}(0, 4.7^2)$                                            \\ \midrule
\multirow{2}{*}{Attack Success Rate}                              & HSJA                                                      & \multicolumn{1}{c|}{0\%}                                                            & \multicolumn{1}{c|}{0\%}                                                                         & 0\%                         & \multicolumn{1}{c|}{0\%}                                                            & \multicolumn{1}{c|}{0\%}                                                                         & 35\%                                                     \\ \cline{2-8} 
& \begin{tabular}[c]{@{}l@{}}Boundary\\ Attack\end{tabular} & \multicolumn{1}{c|}{0\%}                                                            & \multicolumn{1}{c|}{0\%}                                                                         & 0\%                         & \multicolumn{1}{c|}{0\%}                                                            & \multicolumn{1}{c|}{0\%}                                                                         & 55\%                                                     \\ \midrule
\multirow{2}{*}{\begin{tabular}[c]{@{}l@{}}Queries Needed\\ for Attack\\ (Standard Deviation)\end{tabular}} & HSJA                                                      & \multicolumn{1}{c|}{N/A}                                                            & \multicolumn{1}{c|}{N/A}                                                                         & N/A                         & \multicolumn{1}{c|}{N/A}                                                            & \multicolumn{1}{c|}{N/A}                                                                         & \begin{tabular}[c]{@{}c@{}}23,075\\ (1,973)\end{tabular} \\ \cline{2-8} 
& \begin{tabular}[c]{@{}l@{}}Boundary\\ Attack\end{tabular} & \multicolumn{1}{c|}{N/A}                                                            & \multicolumn{1}{c|}{N/A}                                                                         & N/A                         & \multicolumn{1}{c|}{N/A}                                                            & \multicolumn{1}{c|}{N/A}                                                                         & \begin{tabular}[c]{@{}c@{}}41,073\\ (2,128)\end{tabular} \\ \bottomrule
\end{tabular}}
\end{table*}

\mbox{ }

The results of our evaluation are presented in Table \ref{table:evaluate_oars}. We observe
the following, 
\begin{enumerate}
    \item The magnitude of noise needed to bypass Phase 1 (for both quantization variants) is 
    higher than Blacklight, due to the larger quantization constant employed. 
    \item Salted randomized quantization reduces the noise magnitude returned by OARS.
    This is due to the randomness introduced
    during the quantization process, which may initially flag two similar queries as different. 
    However, this noise magnitude is lower than the magnitude needed to bypass detection.
    \item None of the attacks managed to achieve a successful adversarial example
    within 50,000 queries on Phase 1. While this does not exclude the possibility of
    adversarial attacks, the number of queries needed will significantly increase. 
    For comparison, without any defenses, it is possible to generate an adversarial
    example using HSJA in approx. 20,000--24,000 queries, and boundary attack in 
    approx. 30,000--35,000 queries. 
\end{enumerate}

\subsection{Computation cost} \label{subsection:computation_overhead}


\noindent Our system specifications are as follows:
    Processor: Intel(R) Xeon(R) CPU E5-2620 v4 @ 2.10GHz;
RAM: 256GB;
OS:  Ubuntu 20.04.4 LTS. 
Note that GPU is not required.

\subsubsection*{Storage Overhead}
Most of the storage overhead arises from two sources
\begin{enumerate}
    \item The Bloom filters mantained to faciliate query similarity detection.
    \item The soft-labels cached for sequences of similar queries. 
\end{enumerate}
A maximum of 51 Bloom filters with maximum size of $100,000$ and error rate of $0.05$ were mantained,
and maximum total memory overhead incurred at any given time was approximately 7.50 megabytes. On the other hand,
the soft-label cache was more expensive to maintain, as the soft-labels are represented as an array of
floating-point numbers. For the ImageNet dataset, this amounted to a maximum of 17.9 megabytes at
any given time, and at most 0.8 megabytes for CIFAR10. 
No additional GPU resource is needed to compute the soft-labels, as they can be cached during inference time. 

\subsubsection*{Computation Time} 
On average, the total time taken for each of the 5 query sequences (approx. 30,000 queries each)
described in Section \ref{section:expt_robustness_against_sota} is around 1.0 minutes for CIFAR10, and 2.5 minutes
for ImageNet. 
When the detection framework is executed on a multiprocessor system with four cores, the total processing time is reduced to 0.3 minutes and 1.2 minutes, for CIFAR10 and ImageNet, respectively.

\section{Discussion}

\subsection{On Adaptive Attacks}
A natural question arises: Could an adversary possibly perform an adaptive attack which bypasses the proposed detection? In our empirical evaluations, we observed that while the adaptive attack OARS succeeds in generating adversarial examples, it does so only by introducing a substantial amount of noise, which defeats the purpose of an adversarial attack. We argue that due to the stringent detection criteria employed in Phase 1, any attempt to adaptively bypass the defense would necessitate similarly excessive perturbations, rendering such attacks ineffective in real-world scenarios.

There have been a number of adaptive attacks designed to bypass existing SDs, for example OARS and Certifiable Black-Box Attacks \cite{hong2024certifiableblackboxattacksrandomized}. Some of these methods rely on additional information and resources beyond black-box query access, for example, adversary having knowledge of some model-specific adversarial example distribution (see ``adversarial distribution''  described in Section 3 of \cite{hong2024certifiableblackboxattacksrandomized}). For a black-box model, this is approximated by repeatedly querying a model with similar inputs, which would likely be flagged by existing stateful detection methods. 

\subsection{Comparison With Other Defenses}
Defenses against adversarial attacks are typically categorized into a few main classes: stateful defenses, adversarial training, and inference-time defenses. Stateful defenses (e.g., Blacklight and our approach) aim to restrict adversaries from probing the system in ways that reveal meaningful information about model behavior, often by detecting similar queries.

To the best of our knowledge, Boundary Detection \cite{li2022decision} is the only existing soft-label-based detection method that identifies adversarial intent of a query sequence. If the top-1 and top-2 predicted probabilities are close, the queries are  likely near a decision boundary. However, this criterion may be too late, as the adversary is already close to success and could switch to alternative strategies, such as using different accounts or a surrogate model. In contrast, our method detects suspicious behavior earlier by leveraging temporal correlations in soft-label sequences, which are exhibited early in the query sequence (refer to Figure \ref{fig:shifting labels}).

Adversarial training, in contrast, seeks to harden the model against adversarial inputs by  ``smoothing'' its decision boundaries, usually achieved by training on perturbed \cite{xie2020selftrainingnoisystudentimproves, zhang2024imagenet,erichson2024noisymix} or adversarially generated examples \cite{sui2025isdat, wang2024revisiting}. Inference-time defenses typically mitigate adversarial perturbations by injecting randomness, such as adding noise to input queries \cite{qin2021randomnoisedefensequerybased, cohen2019certified}.

\section{Conclusion}
In this paper, we propose
a two-phase 
framework which leverages both feature similarity and temporal behavior of soft-labels
to identify adversarial query sequences. This is a departure from 
most existing SDs, which rely mainly on query similarity 
to flag out malicious queries. 
We further proposed the notion of ``adversarial attacks'' on similarity functions, 
and highlighted that any similarity function employed must be robust against such attacks. 
We propose that introducing randomness in the similarity
function through \textit{salted randomized quantization} improves its robustness
against such ``adversarial attacks''. 
Finally, our proposed framework has been shown to outperform existing
SOTA adversarial defense methods as it 
correctly identifies adversarial queries with a high
TPR, excludes benign queries with a low FPR,  
and is
robust against adaptive attacks designed
to evade SDs.


\bibliographystyle{ACM-Reference-Format}
\bibliography{ref}

\appendix

\end{document}

\end{document}